\documentclass[12pt]{article}

\usepackage[english]{babel}
\usepackage[T1]{fontenc}
\usepackage[utf8]{inputenc}
\usepackage{graphicx}
\usepackage{amsmath,amsfonts,amssymb,amsbsy,amsthm, subfigure,mathabx}
\usepackage[a4paper]{geometry}
\usepackage{tikz,pgfplots}
\usepackage{bm}
\usepackage{comment}

\makeatletter
\newif\if@restonecol
\makeatother

\usepackage[algo2e,ruled,lined,titlenumbered,commentsnumbered]{algorithm2e}

\theoremstyle{remark}

\theoremstyle{plain}
\newtheorem{thm}{Theorem}
\newtheorem{lem}{Lemma}

\newcommand{\un}[1]{\ensuremath{\underline{#1}}}
\newcommand{\uu}[1]{\ensuremath{\un{\un{#1}}}}
\newcommand{\dime}{d}

\newcommand{\sig}{\ensuremath{\uu{\sigma}}}
\newcommand{\sigH}{\ensuremath{\sig_H}}

\newcommand\strain[1]{\uu{\varepsilon}\left(#1\right)}
\newcommand{\dep}{\ensuremath{\un{u}}}
\newcommand{\depv}{\ensuremath{\un{v}}}
\newcommand{\depH}{\ensuremath{\dep_H}}

\newcommand{\depvH}{\ensuremath{\depv_H}}

\newcommand{\efdep}{\ensuremath{\mathbf{u}}}
\newcommand{\lam}{\ensuremath{\boldsymbol{\lambda}}}
\newcommand{\Lam}{\ensuremath{\boldsymbol{\Lambda}}}

\newcommand{\eft}{\ensuremath{\mathbf{t}}}
\newcommand{\kerK}{\ensuremath{\mathbf{R}}}

\newcommand\shapef{\varphi}
\newcommand\shapev{\un{\boldsymbol{\shapef}}_H}
\newcommand{\stiff}{\ensuremath{\mathbf{K}}}
\newcommand{\schur}{\ensuremath{\mathbf{S}}}
\newcommand{\force}{\ensuremath{\mathbf{f}}}

\newcommand{\pa}{\ensuremath{\mathbf{A}}}
\newcommand{\da}{\ensuremath{\mathbf{B}}}
\newcommand{\tpa}{\ensuremath{\mathbf{\tilde{A}}}}
\newcommand{\tda}{\ensuremath{\mathbf{\tilde{B}}}}
\newcommand{\efDep}{\ensuremath{\mathbf{U}}}
\newcommand{\res}{\ensuremath{\mathbf{r}}}
\newcommand{\bz}{\ensuremath{\mathbf{z}}}
\newcommand{\bw}{\ensuremath{\mathbf{w}}}
\newcommand{\bp}{\ensuremath{\mathbf{p}}}

\newcommand{\tlam}{\ensuremath{\boldsymbol{\tilde{\lambda}}}}
\newcommand{\tdep}{\ensuremath{\mathbf{\tilde{u}}}}
\newcommand{\pdep}{\ensuremath{\dot{\mathbf{u}}}}
\newcommand{\plam}{\ensuremath{\boldsymbol{\dot{\lambda}}}}

\newcommand{\matid}{\ensuremath{\mathbf{{I}}}}

\newcommand{\proj}{\ensuremath{\mathbf{{P}}}}

\newcommand{\domain}{\ensuremath{\Omega}}
\newcommand{\s}{\ensuremath{^{(s)}}}
\newcommand{\sT}{\ensuremath{^{(s)^T}}}

\newcommand\trace{\operatorname{tr}}

\newcommand\hooke{\mathbb{H}}
\newcommand\KA{\ensuremath{\mathrm{KA}}}
\newcommand\KAo{\ensuremath{\KA^0}}
\newcommand\KAH{\ensuremath{\KA_H}}

\newcommand\SA{\ensuremath{\mathrm{SA}}}

\newcommand{\broken}{\ensuremath{\KA(\bigcup\domain\s)}}

\newcommand\ecr[2]{\ensuremath{\mathrm{e_{CR_{#2}}(#1)}}}
\newcommand\ecrc[2]{\ensuremath{\mathrm{e^2_{CR_{#2}}(#1)}}}
\newcommand\enernorm[2]{\|#1\|_{\hooke^{-1},#2}}

\title{A strict error bound with separated contributions of the discretization and of the iterative solver in non-overlapping domain decomposition methods}
\author{Valentine Rey\footnote{valentine.rey@lmt.ens-cachan.fr}, Christian Rey\footnote{christian.rey@lmt.ens-cachan.fr}, Pierre Gosselet{gosselet@lmt.ens-cachan.fr} \\
LMT-Cachan / ENS-Cachan, CNRS, UPMC, Pres UniverSud Paris \\
61, avenue du pr\'esident Wilson, 94235 Cachan, France
}

\begin{document}
\maketitle
\begin{abstract}
This paper deals with the estimation of the distance between the solution of a static linear mechanic problem and its approximation by the finite element method solved with a non-overlapping domain decomposition method (FETI or BDD). We propose a new strict upper bound of the error which separates the contribution of the iterative solver and the contribution of the discretization. Numerical assessments show that the bound is sharp and enables us to define an objective stopping criterion for the iterative solver.

{\bf Keywords:} Verification; domain decomposition methods; FETI; BDD; convergence criterion.
\end{abstract}

\section{Introduction}
Developing robust numerical methods to solve systems of partial differential equations has become a major challenge in engineering. Indeed industrialists wish to adopt virtual testing in order to replace expensive experimental studies up to the certification of their structures. The massive use of virtual prototyping relies on the capacity to warrant the quality of the numerical solutions. In the context of the Finite Element Method (FEM), \textit{a posteriori} error estimators permit to estimate the distance between the unknown exact solution and the numerical solution. 
Initial methods \cite{Bab78a,Lad83,Zie87} evaluated globally the effect of the spatial discretization for linear problems, they have been extended to non-linear and time-dependent problems, and to the estimation of the error on quantities of interest. Another prerequisite for virtual testing is the ability to conduct large scale computations, because reliable models involve lots of degrees of freedom. Non-overlapping Domain Decomposition Methods (DDM) offer a favorable framework for fast iterative solvers adapted to modern clusters \cite{Gos06}.

Most upper bounds for the error which do not involve constants rely on the computation of admissible stress and displacement fields. In a recent paper \cite{Par10}, the classical methods to construct statically admissible fields were extended to the framework of substructured problems. The estimator which ensues is fully parallel and totally integrated to classical DD solvers BDD \cite{Let94} and FETI \cite{Far94bis}. It provides a guaranteed upper bound whether the iterative solver of the interface problem has converged or not; unfortunately, it is not able to separate the different sources of error, namely the error due to the discretization and the error due to the lack of convergence at the interface. 
In \cite{Par10}, a Gamma-shape structure clamped on its basis and sollicitated in traction and shear on its upper-right side was considered. The structure was split into $8$ subdomains (Figure~\ref{fig:courbe_La}) and the problem was solved with classical DD solver. The error estimation provided by the estimator proposed in \cite{Par10} is computed at each iteration of the DD solver (Figure~\ref{fig:courbe_Lb}).
We clearly observe  L-shaped curves highlighting the fast convergence rate of the estimator $e_{CR}^{DDM}$ with respect to the domain decomposition residual. 
After a few iterations, the error due to the non-verification of the continuity and balance on the interface is insignificant compared to the contribution of the discretization error (which can be visualized by the estimator of the sequential problem). On that example, we observe that, whatever the substructuring, after 4 iterations there is no improvement of the approximation, while classical stopping criteria based on the decrease of a norm of the interface residual would imply 5 times more iterations.
\begin{figure}[ht]\hfill
\subfigure[Domain decomposition and loading]{\centering
 \includegraphics[width=0.25\textwidth]{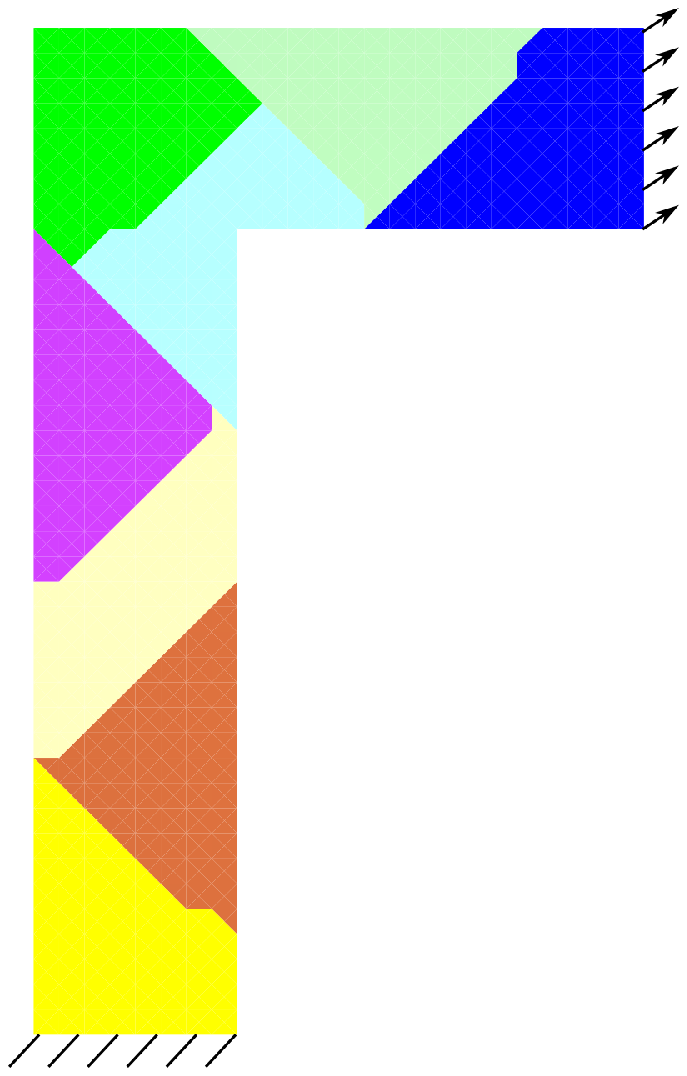}\label{fig:courbe_La}}
\subfigure[Convergence of error estimator vs DD residual]{
\centering \includegraphics[width=.7\textwidth]{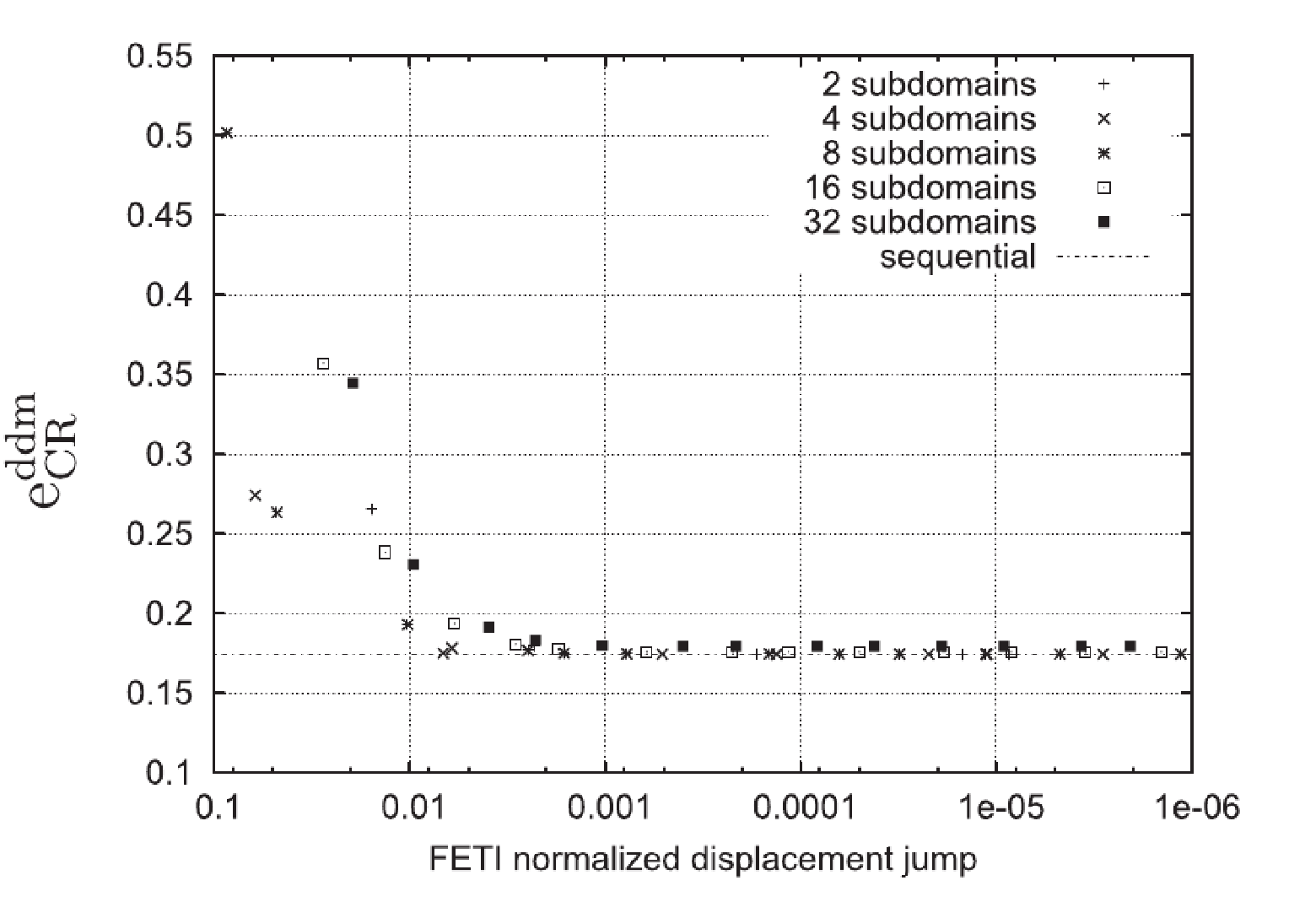}\label{fig:courbe_Lb}}
\caption{Error estimator from \cite{Par10}}
\end{figure}

Then in order to avoid oversolving, we wish to distinguish the contributions of the discretization and of the iterative solver to the estimation of the error. The non-convergence of the solver will be referred to as the algebraic error and no other sources of errors (rounding, representation of the loadings) will be considered.


Various articles have dealt with the separation of the contributions to the error and with the definition of new stopping criteria. In~\cite{Ari04} the author insists on the use of the energy norm for the measurement of the residual instead of a classical Euclidean norm, in order to link the iterative methods to the properties of the approximated problem. In the framework of multigrid methods, an adaptive procedure to define both the refinement and the stopping criterion is developed in~\cite{Bec95}. However, this technique demands the computation of constants since it is based on \textit{a priori} estimates. The error in constitutive relation does not require the calculation of such constants and provides guaranteed upper bounds, it was applied to various problems that introduce other sources of error. In~\cite{Gal96,Gal07}, the authors define a time error indicator to separate the part of the error due to the time discretization from the part due to the space discretization, and they use it to optimize the time steps. This work is extended in~\cite{Lad98qua} in which an indicator of the effect of non-linear iterations complete the total error estimation. Finally, in the case of contact problems, the separation of the discretization error and of the algebraic error is performed for problems solved with the fixed-point method \cite{Lou03bis} and with a Neumann-Dirichlet algorithm~\cite{Gal10}. Nevertheless, the computation of all those indicators requires the resolution of auxiliary problems which considerably increase the cost of the estimation (indeed, one has to compute statically admissible fields for each problem, which can be a costly step).
For the finite volume method, the separation of the different sources of error and the definition of a new stopping criterion is exposed in~\cite{Jir10} for second-order elliptic problems. 
This question of balancing the sources of error has also been addressed for error estimation on quantity of interest: in~\cite{Mei09}, a goal-oriented procedure to solve a problem with the multigrid method with a level of precision specified by the user is proposed; a similar approach is presented in~\cite{Choi04} for the bound method~\cite{Par97,Par98}. 

In this paper, we present a new guaranteed upper bound that separates the algebraic error (represented by a well chosen norm of the residual) from the discretization error of the subdomains in the case of a linear problem solved with a classical DD solver (BDD or FETI). 
This separation enables us to define a new stopping criterion for the iterative solver. The method relies on the parallel procedures to build admissible fields proposed in \cite{Par10} but compared to the estimator of that paper, the procedures are called much less often.

The paper is organized as follow. In Section~\ref{sec:rappels} we define the reference problem, we recall the principle of the error in constitutive relation and we detail the error estimation in the substructured context by highlighting the fields created in the FETI and BDD solvers. In Section~\ref{sec:theoreme} we prove the new guaranteed upper bound separating the algebraic error and the discretization error on each subdomain and we explain how it leads us to define a new stopping criterion for the iterative solver. In Section~\ref{sec:assessment}, we apply our new upper bound to two 2D mechanical problems. We compare the behavior of the new upper bound with the one introduced in \cite{Par10} and study the evolution of each term of the inequality during the iterations. Section~\ref{sec:conclusion} concludes the paper.

\section{A posteriori error estimation in substructured context}
\label{sec:rappels}
\subsection{Reference problem}
Let $\mathbb{R}^\dime$ represent the physical space
($d$ is the dimension of the physical space).
Let us consider the static equilibrium  of a (polyhedral) structure which occupies the open domain $\domain\subset\mathbb{R}^\dime$ and which is subjected to  given body force $\un{f}$ within $\Omega$,  to given
traction force $\un{g}$ on $\partial_g\Omega$  and to given displacement field $\dep_d$ on the  complementary part of the boundary $\partial_u\Omega$ ($\operatorname{meas}(\partial_u\Omega)\neq 0$). We  assume that the structure undergoes  small   perturbations  and  that  the  material   is  linear  elastic, characterized by Hooke's  elasticity tensor $\hooke$.  Let $\dep$ be  the unknown displacement field, $\strain{\dep}$ the symmetric part  of the gradient of $\dep$, $\sig$ the Cauchy stress tensor. Let $\omega$ be an open subset of $\domain$.

We introduce two affine subspaces and one positive form:
\begin{itemize}
\item Affine subspace of kinematic admissible fields (KA-fields)
\begin{equation}\label{eq:KA}
  \KA(\omega)=\left\{ \dep\in \left(\mathtt{H}^1(\omega)\right)^\dime,\ \dep = \dep_d \text{ on }\partial\omega\bigcap\partial_u\domain \right\}
\end{equation}
and we note $\KAo$ the following linear subspace:
\begin{equation}\label{eq:KAo}
  \KAo(\omega)=\left\{ \dep\in \left(\mathtt{H}^1(\omega)\right)^\dime,\ \dep = 0 \text{ on }\partial\omega\setminus\partial_g\domain \right\}
\end{equation}
\item Affine subspace of statically admissible fields (SA-fields)
\begin{multline}\label{eq:SA}
  \SA(\omega)
  =\Bigg\lbrace   \uu{\tau}\in  \left(\mathtt{L}^2(\omega)\right)^{\dime\times \dime}_{\text{sym}}; \
    \forall  \depv \in  \KAo(\omega),\ \\ \int\limits_\omega
  \uu{\tau}:\strain{\depv}    d\domain    =    \int\limits_\omega   \un{f} \cdot\depv    d\domain +
  \int\limits_{\partial_g\domain\bigcap\partial\omega} \un{g}\cdot\depv dS   \Bigg\rbrace
\end{multline}
\item Error in constitutive equation 
\begin{equation}\label{eq:ecr}
  \ecr{\dep,\sig}{\omega}= \enernorm{\sig-\hooke:\strain{\dep}}{\omega}
\end{equation}
where ${\enernorm{\uu{x}}{\omega}}=\displaystyle \sqrt{\int_\omega \left( \uu{x}: {\hooke}^{-1}:\uu{x} \right)d\domain}$
\end{itemize}

The mechanical problem set on $\domain$ can be formulated as:
\begin{equation}\label{eq:refpb}
  \text{Find } \left(\dep_{ex},\sig_{ex}\right)\in\KA(\domain)\times\SA(\domain) \text{ such  that } \ecr{\dep_{ex},\sig_{ex}}{\domain}=0
\end{equation}
The solution to this problem, named ``exact'' solution, exists and is unique.

\subsubsection{Finite element approximation}
Let us consider a tessellation  of ${\domain}$ to which we  associate the finite-dimensional  subspace  $\KAH(\domain)$  of $\KA(\domain)$.  The  classical finite element displacement approximation consists in searching:
\begin{equation}
  \begin{aligned}
    \depH&\in\KAH(\domain)\\
    \sigH&=\hooke:\strain{\depH}   \\
    \int_{\domain}
  \sigH:\strain{\depvH}   d\domain  &=   \int_{\domain} \un{f}\cdot\depvH   d\domain   +
  \int_{\partial_g\domain} \un{g}\cdot\depvH dS,\qquad \forall\depvH\in \KAo_H(\domain)
    \end{aligned}
\end{equation}
Of course the approximation is due to the fact that in most cases $\sigH\notin\SA(\domain)$.

After introducing the matrix $\shapev$ of  shape functions which form a basis of $\KAH(\domain)$ (extended to Dirichlet degrees of freedom)  and the  vector of  nodal unknowns  $\efdep$ so that $\depH=\shapev \efdep$, the  classical finite  element method leads  to the linear
system:
\begin{equation}\label{eq:globalFE}
\begin{pmatrix} \stiff_{rr} & \stiff_{rd} \\ \stiff_{dr} & \stiff_{dd}
\end{pmatrix}\begin{pmatrix}\efdep_r\\ \efdep_d\end{pmatrix} = \begin{pmatrix}\force_r\\\force_d \end{pmatrix} +\begin{pmatrix} 0\\\lam_d \end{pmatrix}
\end{equation}
where $\stiff$ is  the (symmetric positive definite) stiffness  matrix 
and $\force$ is the vector of generalized forces; Subscript $d$ stands for Dirichlet degrees of freedom (where displacements are prescribed) and Subscript $r$ represents the remaining degrees of freedom so that unknowns are $\efdep_r$ and $\lam_d$ where Vector $\lam_d$ represents the nodal reactions:
\begin{equation}\label{eq:lambdad}
\lam_d^T=
    \int_{\domain}
  \sigH:\strain{{\shapev}_d}   d\domain - \int_{\domain} \un{f}\cdot{\shapev}_d  d\domain -
  \int_{\partial_g\domain} \un{g}\cdot{\shapev}_d dS
\end{equation}
where ${\shapev}_d$ is the matrix of shape functions restricted to the Dirichlet nodes and $\un{n}$ the outer normal vector.

\subsubsection{Error estimation}
The estimator we choose is based on the error in constitutive equation. The fundamental relation is the following (Prager-Synge theorm, see for instance \cite{Lad04}):
\begin{equation}\label{eq:erdc}
\begin{split}
&\forall (\hat{\dep},\hat{\sig})\in\KA(\Omega)\times\SA(\Omega),\\ 
&\left\|\strain{\dep_{ex}}-\strain{\hat{\dep}}\right\|_{\hooke,\Omega}^2 + \left\|\sig_{ex}-\hat{\sig}\right\|_{\hooke^{-1},\Omega}^2 =  \ecrc{\hat{\dep},\hat{\sig}}{\Omega}
\end{split}
\end{equation}
We note $\vvvert \depv \vvvert_\domain=\left\|\strain{\depv}\right\|_{\hooke,\Omega} $ the energy norm of the displacement, and we retain the following bound:
\begin{equation}\label{eq:erdc2}
 \vvvert  \dep_{ex} -\hat{\dep}  \vvvert_\domain  \leqslant  \ecr{\hat{\dep},\hat{\sig}}{\Omega}
\end{equation}
For this problem one can choose $\hat{\dep}=\depH\in\KA(\domain)$. The construction of $\hat{\sig}\in\SA(\domain)$ is a more complex problem solved by various approaches \cite{Lad91,Par09,Ple11}.

\subsection{A posteriori error estimation in substructured context}
\subsubsection{Substructured formulation}
 Let us consider a decomposition of domain $\Omega$ in $N_{sd}$ open subsets $\Omega^{(s)}$ such that $\Omega^{(s)}\bigcap \Omega^{(s')}=\emptyset$ for $s\neq s'$ and ${\Omega}=\bigcup_s {\Omega}^{(s)}$. The interface between subdomains $\Gamma^{(s,s')}={\Omega}^{(s)}\bigcap {\Omega}^{(s')}$ is supposed to be regular enough for traces 
of locally admissible fields
to be well defined.
 
The mechanical problem on the substructured configuration writes :
\begin{equation}
  \forall s \left\{
   \begin{aligned}
&\dep^{(s)} \in \KA(\Omega^{(s)})  \\
& \sig^{(s)} \in \SA(\Omega^{(s)}) \\
&  \ecr{\dep^{(s)},\sig^{(s)}}{\Omega^{(s)}}=0
\end{aligned}  \right. \text{ and } \forall (s,s')\left\{
    \begin{aligned}
   & \trace(\dep^{(s)})=\trace(\dep^{(s')}) \text{ on } \Gamma^{(s,s')}\\
 &  \sig^{(s)} \cdot\underline{n}^{(s)} +\sig^{(s')}\cdot \underline{n}^{(s')} =\underline{0} \text{ on } \Gamma^{(s,s')}
    \end{aligned} 
  \right.
\end{equation}
Indeed, the kinematic and static admissibility of each $(\dep^{(s)},\sig^{(s)})$ inside $\Omega^{(s)}$ is not sufficient to be globally admissible. The displacements need to be continuous and the tractions need to be balanced on interfaces.
The set of fields $\dep$ defined on $\domain$ such that $\dep_{|\domain\s}\in\KA(\Omega^{(s)})$ without interface continuity is a broken space which we note \broken.

\subsubsection{Finite element approximation for the substructured problem}
We assume that the tessellation of  ${\domain}$ and the substructuring are conforming. This hypothesis implies that each element only belongs to one subdomain and nodes are matching on the interfaces. Each degree of freedom is either located inside a subdomain (Subscript i) or on its boundary (Subscript b).

Let  $\eft^{(s)}$ be the discrete trace operator, so that $\efdep_b\s=\eft^{(s)}\efdep\s$. Let us introduce the unknown nodal reaction on the interface $\lam\s$, the equilibrium of each subdomain writes:
\begin{equation}\label{eq:efddeq}
\stiff\s \efdep\s = \force\s + {\eft\s}^T \lam\s
\end{equation}
Let $(\pa\s)$ and $(\da\s)$ be the primal and dual assembly operator so that the discrete counterpart of the interface admissibility equations is:
\begin{equation}\label{eq:efintadmiss}
\left\{\begin{aligned}
\sum_s\da\s \eft\s\efdep\s &=0 \\
\sum_s\pa\s \lam\s &=0 \\
\end{aligned}\right.
\end{equation}
Equations (\ref{eq:efddeq},\ref{eq:efintadmiss}) form the discrete substructured system, which is equivalent to the global problem \eqref{eq:globalFE}.

Note that in the case when Subdomain $s$ has not enough Dirichlet boundary conditions then the  reaction has to balance the subdomain with respect to rigid body motions. Let $\kerK\s$ be a basis of $\operatorname{ker}(\stiff\s)$, we have:
\begin{equation}\label{eq:ortho_modes_rigides}
{\kerK\s}^T(\force\s + {\eft\s}^T \lam\s) =0
\end{equation}

\subsubsection{Domain decomposition solvers}
Classical BDD and FETI solvers are well described in many papers (for instance, see~\cite{Gos06} and the associated bibliography). 
Very briefly, BDD and FETI are Krylov iterative solvers for the reformulation of the problem in term interface quantities: in BDD, the continuous interface displacement which nullifies the interface lack of balance is sought, whereas in FETI the balanced nodal reaction field which nullifies the displacement gap is searched for.
We emphasize on the fields which are created along the algorithms and which are useful for error estimation.

We name two important parallel procedures which are used in these methods and which correspond to solving local problems with Dirichlet conditions on the interface (subscript D) and local problems with Neumann conditions on the interface (subscript N):\smallskip

\vline
\begin{minipage}[t]{.45\textwidth}
$(\lam\s_D,\efdep_D\s)=\mathtt{Solve}_D(\efdep\s_b,\force\s)$:
\begin{equation*}
\left\{ \begin{aligned}
&\stiff\s \efdep_D\s = \force\s + {\eft\s}^T \lam_D\s \\
&\eft\s\efdep_D\s=\efdep\s_b
\end{aligned}
\right.
\end{equation*}
\end{minipage}\hfill\vline
\begin{minipage}[t]{.45\textwidth}
$(\efdep_N\s)=\mathtt{Solve}_N(\lam_N\s,\force\s)$
\begin{equation*}
\left\{ \begin{aligned}
&\stiff\s \efdep_N\s = \force\s + {\eft\s}^T \lam_N\s\\
&\text{where }(\lam_N\s)_s \text{ satisfy Eq.~} \eqref{eq:ortho_modes_rigides}
\end{aligned}
\right.
\end{equation*}

\end{minipage}

\medskip

When developing these methods, we get:
\begin{equation*}
\begin{aligned}
&(\efdep_D\s)_i = {\stiff\s_{ii}}^{-1}\left( \force_i\s - \stiff\s_{ib} \efdep_b\s\right)\qquad\text{and}\qquad (\efdep_D\s)_b= \efdep_b\s\\
&\lam\s_D = \schur\s \efdep_b\s -\force\s_b+\stiff\s_{bi} {\stiff\s_{ii}}^{-1}\force\s_i\\
&\efdep_N\s = {\stiff\s}^+\left( \force\s + {\eft\s}^T \lam_N\s\right)
\end{aligned}
\end{equation*}
where $\schur\s=\left(\stiff\s_{bb}-\stiff\s_{bi} {\stiff\s_{ii}}^{-1}\stiff\s_{ib}\right)$ is the Schur complement matrix, and ${\stiff\s}^+$ is a pseudo-inverse of ${\stiff\s}$.
Depending on the method (BDD or FETI), the way to ensure the well-posedness of Neumann  problems varies. In general the implementation relies on initialization and projection, leading to a  preconditioned projected conjugate gradient, we then introduce the generic method $\mathtt{Initialize}$ and the projectors $\proj_1$ and $\proj_2$ in our algorithms (see~\cite{Gos06} for details).
 
Other important ingredients are the scaled assembling operators $(\tpa\s)$ and $(\tda\s)$. These matrices are any solutions to the following equations:
\begin{equation}
\sum_s \pa\s \tpa\sT = \matid \qquad \text{and}\qquad \sum_s \da\s \tda\sT = \matid 
\end{equation}
See~\cite{Rix99bis,Kla01} for classical definitions of these operators.

In algorithms \ref{alg:bdd} and \ref{alg:feti}, we emphasize the fields which are specifically rebuilt for error estimation using right-aligned C-style  comments.
\begin{algorithm2e}[th]\caption{BDD: main unknown $\efDep$ on the interface}\label{alg:bdd}
$\efDep=\mathtt{Initialize}(\force\s)$ \;  %
$(\lam\s_D,\efdep_D\s)=\mathtt{Solve}_D(\pa\sT\efDep,\force\s)$\;
Compute residual $\res=\sum_s\pa\s\lam\s_D$\;
Define local traction $\tlam\s=\tpa\sT\res$ \tcp*[r]{$\lam_N\s=\lam\s_D-\tlam\s$} 
$\tdep\s=\mathtt{Solve}_N(\tlam\s,0)$ \tcp*[r]{$\efdep_N\s=\efdep_D\s-{\tdep}\s $}
Preconditioned residual $\bz=\sum_s\tpa\s\tdep\s$ \;
Search direction $\bw= \proj_1 \bz$\;
\While{$\sqrt{\res^T\bz}>\epsilon$}{%
  $(\plam\s_D,\pdep_D\s)=\mathtt{Solve}_D(\pa\sT \bw,0)$\;
  $\bp=\sum_s\pa\s\plam\s_D$\;
  $\alpha=(\res^T\bz)/(\bp^T\bw)$\;
  $\efDep\leftarrow \efDep+\alpha \bw$  \tcp*[r]{$\begin{aligned}&\efdep_D\s\leftarrow \efdep_D\s+\alpha \pdep_D\s \\&\lam_D\s\leftarrow \lam_D\s+\alpha \plam_D\s\end{aligned} $}
  $\res \leftarrow \res-\alpha \bp$\;
  $\tlam\s=\tpa\sT\res$ \tcp*[r]{$\lam_N\s=\lam\s_D-\tlam\s$}
  $\tdep\s=\mathtt{Solve}_N(\tlam\s,0)$ \tcp*[r]{$\efdep_N\s=\efdep_D\s-{\tdep}\s $}
  $\bz = \sum_s\tpa\s\tdep_b\s$\;
  $\bw \leftarrow \proj_1 \bz - (\bp^T\proj_1\bz)/(\bp^T\bw) \bw $
}%
\end{algorithm2e}
\begin{algorithm2e}[th]\caption{FETI: main unknown $\Lam$}\label{alg:feti}
$\Lam=\mathtt{Initialize}(\force\s)$ \;  %
Local reactions $\lam_N\s=\da\sT\Lam$\;
$(\efdep_N\s)=\mathtt{Solve}_N(\lam_N\s,\force\s)$\;
Compute residual $\res=\proj_2^T(\sum_s\da\s\eft\s\efdep\s_N)$\;
Define local displacement $\tdep\s_b=\tda\sT\res$; \;
$(\tlam\s,\tdep\s)=\mathtt{Solve}_D(\tdep_b\s,0)$  \tcp*[r]{$\begin{aligned}&\efdep_D\s=\efdep_N\s-{\tdep}\s \\&\lam_D\s=\lam_N\s-\tlam\s\end{aligned}$}
Preconditioned residual $\bz=\proj_2(\sum_s\tda\s\tlam\s)$ \;
Search direction $\bw=\bz$\;
\While{$\sqrt{\res^T\bz}>\epsilon$}{%
  $(\pdep_N\s)=\mathtt{Solve}_N(\da\sT \bw,0)$\;
  $\bp=\proj_2^T(\sum_s\da\s\eft\s\pdep\s_N)$\;
  $\alpha=(\res^T\bz)/(\bp^T\bw)$\;
  $\Lam\leftarrow \Lam+\alpha \bw$  \tcp*[r]{$\begin{aligned}&\efdep_N\s\leftarrow \efdep_N\s+\alpha \pdep_N\s\\&\lam_N\s=\da\sT\Lam\end{aligned}$}
  $\res \leftarrow \res-\alpha \bp$\;
  $\tdep_b\s=\tda\sT\res$\;
  $(\tlam\s,\tdep\s)=\mathtt{Solve}_D(\tdep_b\s,0)$ \tcp*[r]{$\begin{aligned}&\efdep_D\s=\efdep_N\s-{\tdep}\s\\& \lam_D\s=\lam_N\s-\tlam\s \end{aligned}$}
  $\bz = \proj_2(\sum_s\tda\s\tlam\s)$\;
  $\bw \leftarrow \bz - (\bp^T\bz)/(\bp^T\bw) \bw $
}%
\end{algorithm2e}

\subsubsection{A posteriori error estimator}\label{ssec:apee}

In order to apply formula \eqref{eq:erdc2}, \cite{Par10} proposed a parallel procedure to build admissible fields. Indeed  the BDD and FETI solvers provide at every iteration the following vectors :
\begin{itemize}
\item[$(\efdep_D\s)_s$:] displacement vectors which are continuous at the interface so that the field $(\dep_D\s)_s=(\shapev\s\efdep_D\s)_s\in \KA(\domain)$, $(\lam_D\s)$ are the nodal reaction associated to that Dirichlet condition, they are not balanced before convergence.
\item[$(\efdep_N\s)_s$:] displacement vectors associated to nodal Reactions $(\lam_N\s)_s$ which are balanced at the interface. Displacement field $(\dep_N\s)_s=
(\shapev\s\efdep_N\s)_s \in \broken$ (hence there is no continuity across interfaces) and the associated stress field $\sig_N\s = \hooke:\strain{\shapev\s\efdep_N\s}$ can be employed (with additional input $\lam_N\s$) to build stress fields $ \hat{\sig}_N\s$ which are statically admissible $(\hat{\sig}_N\s)_s\in\SA(\Omega)$.
\end{itemize}

The following estimator is then fully parallel:
\begin{equation}\label{eq:estimgus}
 \vvvert  \dep_{ex} -\dep_D \vvvert_\domain^2 = \sum_s \vvvert  \dep_{ex}\s -\dep_D\s \vvvert_{\domain\s}^2\leqslant \sum_s\ecrc{\dep_D\s,\hat{\sig}_N\s}{\Omega\s}
\end{equation}

\section{New error bound}
\label{sec:theoreme}

Our objective is to separate the contributions to the global error of the discretization error and of the algebraic error, so that a new stopping criterion for the iterative solver can be defined, avoiding useless iterations that cannot reduce the global error (see Figure~\ref{fig:courbe_Lb}).\medskip

First, we exhibit a new bound which, compared to~\eqref{eq:estimgus}, involves   $\un{u}_N$ instead of $\un{u}_D$, this upper bound of $ \vvvert  \un{u}_{ex}-\un{u}_N  \vvvert _{\Omega} $  distinguishes the algebraic error due to the use of the DD iterative solver from the discretization error due to the finite element approximation.

The fundamental result is the following theorem:
\begin{thm}\label{thm:separation}
Using the notation of algorithms~\ref{alg:bdd} or~\ref{alg:feti}, we have
\begin{equation}\label{eq:erralg}
 \vvvert  \un{u}_{ex}-\un{u}_N  \vvvert _{\Omega} \leq  \sqrt{\res^T\bz} + \sqrt{\sum_s\ecrc{\dep_N\s,\hat{\sig}_N\s}{\Omega\s}}
\end{equation}
\end{thm}
In other words, at each iteration of the solver, the distance between the displacement field $\un{u}_N$ and the exact solution is bounded by the preconditioner norm of the residual of the conjugate gradient solver plus a sum over subdomains of errors in constitutive relation which depend on the iteration and on the discretization and which are computed in parallel.

The preconditioner-norm of the residual is a purely algebraic quantity computed during the conjugate gradient iterations,  it can be used in order to tell the convergence of the algorithm. In the assessments, we will verify that the term $\sqrt{\sum_s\ecrc{\dep_N\s,\hat{\sig}_N\s}{\Omega\s}}$ is mostly driven by the discretization, so that it varies slightly during the iterations.\medskip

The proof of the theorem is based on two lemmas.
\begin{lem}\label{thm:separation_continue}
Let $\un{u}_{ex}\in\KA(\Omega)$ solve the reference problem \eqref{eq:refpb}, $(\un{u}_D\s)_s\in\KA(\Omega)$ and $(\un{u}_N\s)_s\in\broken$ be as defined in Section~\ref{ssec:apee}, then:
\begin{equation}\label{eq:lemma1}
 \vvvert  \un{u}_{ex}-\un{u}_N  \vvvert _{\Omega} \leq  \vvvert  \un{u}_N-\un{u}_D  \vvvert_{\Omega} + \sqrt{\sum_s\ecrc{\dep_N\s,\hat{\sig}_N\s}{\Omega\s}}
\end{equation}
\end{lem}
\begin{proof}
The proof of this lemma is based on the following result, proved in \cite{vohralik_posteriori_2007}:
 for $(\un{u},\depv) \in \KA(\Omega)\times\KA(\Omega)$ and $\widetilde{\un{u}} \in \broken$, 
\begin{equation}\label{eq:lemvohralik}
 \vvvert \un{u}-\widetilde{\un{u}} \vvvert_\Omega \leq \vvvert \depv-\widetilde{\un{u}} \vvvert_\Omega +
 \sum_s \int_{\Omega_s} \strain{\un{u}-\widetilde{\un{u}}}: \hooke :\strain{\frac{\un{u}-\depv}{\vvvert \un{u}-\depv \vvvert_\Omega}} d\Omega
 \end{equation}
We apply this result with $\un{u}=\un{u}_{ex}$ and $\widetilde{\un{u}}=(\un{u}_N\s)_s$, we note $\un{\varphi} = \frac{\un{u}_{ex}-\depv}{\vvvert \un{u}_{ex}-\depv\vvvert_\Omega}\in \KAo(\Omega)$, the second term of the expression is simplified by the introduction of 
$(\hat{\sig}_N\s)_s\in\SA(\Omega)$:
\begin{equation}
\begin{aligned}
\sum_s \int_{\Omega_s} \strain{\un{u}_{ex}\s-\un{u}_N\s} :\hooke :\strain{\un{\varphi}\s} d\Omega &=
 \sum_s \int_{\Omega_s} (\sig_{ex}\s-\sig_N\s ) :\strain{\un{\varphi}\s} d\Omega \\
& = \sum_s \int_{\Omega_s} (\hat{\sig}_N\s-\sig_N\s ) :\strain{\un{\varphi}\s} d\Omega \\
& \leqslant \sum_s \lVert \hat{\sig}_N\s-\sig_N\s \rVert_{\hooke^{-1},\Omega\s} \vvvert \un{\varphi}\s\vvvert_{\Omega\s}\\
& \leqslant \sqrt{\sum_s \lVert \hat{\sig}_N\s-\sig_N\s \rVert_{\hooke^{-1},\Omega\s}^2}\\&\qquad\qquad\qquad\qquad\times \underbrace{\sqrt{\sum_s\vvvert \un{\varphi}\s\vvvert_{\Omega\s}^2}}_{=1}
\end{aligned}
\end{equation}
where we have used Cauchy-Schwarz inequality (two times). 
By definition  $\lVert \hat{\sig}_N\s-\sig_N\s \rVert_{\hooke^{-1},\Omega\s} = \ecr{\dep_N\s,\hat{\sig}_N\s}{\Omega\s}$. One just has to choose $\depv=\un{u}_D$ in the first term of  Equation~\eqref{eq:lemvohralik}.
\end{proof}

\begin{lem} Using the notation of algorithms~\ref{alg:bdd} or~\ref{alg:feti}, we have: 
\begin{equation}\label{eq:lemma2}
\lVert \un{u}_N-\un{u}_D \rVert_{\hooke,\Omega}^2 = \res^T\bz
\end{equation}
in other words, the distance between $\un{u}_N$ and $\un{u}_D$ in the energy norm is the pre\-con\-di\-tio\-ner norm of the residual.
\end{lem}
\begin{proof}
Let us transform the expression of the first term using Stokes' theorem and definition~\eqref{eq:lambdad}:
\begin{equation}
\begin{aligned}
\lVert \un{u}_N-\un{u}_D \rVert_{\hooke,\Omega}^2 &= \sum_s \int_{\Omega\s} \strain{\un{u}_D\s-\un{u}_N\s} : \hooke: \strain{\un{u}_D\s-\un{u}_N\s} d\Omega\\
&= 
 \sum_s \left(\lam_D\s-\lam_N\s\right)^T{\eft\s} \left(\efdep_D\s-\efdep_N\s\right)\\
 &=\sum_s\tlam\sT\tdep_b\s
\end{aligned}
\end{equation}
We now need to make particular cases depending on the algorithm. 
 In the  BDD case, we have:
\begin{equation}
\begin{aligned}
 \sum_s\tlam\sT\tdep_b\s = \sum_s (\tpa\sT \res)^T \tdep_b\s= \res^T \bz
\end{aligned}
\end{equation}
In the FETI case, we use the fact that $\res=\proj_2^T\res$:
\begin{equation}
\begin{aligned}
\sum_s\tlam\sT\tdep_b\s = \sum_s\tlam\sT \tda\sT \res= \sum_s\tlam\sT \tda\sT \proj_2^T\res = \res^T \bz
\end{aligned}
\end{equation}
\end{proof}

This leads to the resolution strategy of Algorithm~\ref{alg:ddstop}. The idea is to iterate until the algebraic error is negligible with respect to the estimation of the discretization error. 
\begin{algorithm2e}[th]\caption{DD solver with adapted stopping criterion}\label{alg:ddstop}
Set $\alpha>1$ and $\beta\geqslant 1$\;
Initialize, get $(\efdep_N\s,\lam_N\s,\res,\bz)$\;%
Estimate discretization error $e^2=\sum_s\ecrc{\dep_N\s,\hat{\sig}_N\s}{\Omega\s}$ \;
\While{$\sqrt{\res^T\bz}>e/\alpha$}{
\While{$\sqrt{\res^T\bz}>e/(\alpha \beta)$}{
Make Alg.~\ref{alg:bdd} or Alg.~\ref{alg:feti} iterations, get $(\efdep_N\s,\lam_N\s)$
}
Update error estimator $e^2=\sum_s\ecrc{\dep_N\s,\hat{\sig}_N\s}{\Omega\s}$ \;
}
\end{algorithm2e}
The objective being to have the norm of the residual $\alpha$ times smaller than the estimation of the discretization error ($\alpha=10$ is a typical value), the coefficient $\beta$ (typically 2) takes into account the small variation of the estimate $e$ of the discretization error along iterations.\\


\noindent
{\bf Remark:}  
As engineers often expect the displacement field to be continuous, we can derive a bound for the error associated with the continuous field $\un{u}_D$ using the triangular inequality and \eqref{eq:erralg}-\eqref{eq:lemma2}:
\begin{equation}
 \vvvert  \un{u}_{ex}-\un{u}_D  \vvvert _{\Omega} \leq  2 \, \sqrt{\res^T\bz} + \sqrt{\sum_s\ecrc{\dep_N\s,\hat{\sig}_N\s}{\Omega\s}}
\end{equation}

\section{Numerical assessment}
\label{sec:assessment}
In this section, we present the new upper bound obtained by applying  Theorem~\ref{thm:separation_continue} for two mechanical problems.The FE computation and the construction of statically admissible stress fields are performed using an Octave code. The material is chosen to be isotropic, homogeneous, linear and elastic with Young's modulus $E= 1\,$Pa and Poisson's ratio $\nu=0.3$. 

For each case, we computed the error estimation for the sequential problem and for the substructured problem (as it is proposed in \cite{Par10}) and the two sources of error provided by the application of Theorem~\ref{thm:separation_continue}.

\subsection{Rectangular domain with known solution}

First, let us consider a rectangular structure $\Omega=[0;8]\times[0;1]$. Homogeneous Dirichlet boundary conditions are considered on all the boundary and the domain is subjected to a polynomial body force such that the exact solution is known : 
\begin{equation*}
\underline{u}_{ex} =x(x-8l)y(y-l)^3 \underline{e}_x + x y^2 (x-8) (y-1)\underline{e}_y
\end{equation*}

Therefore, we are able to compute the true errors $\vvvert  \un{u}_{ex}-\un{u}_N  \vvvert _{\Omega}$ and $\vvvert  \un{u}_{ex}-\un{u}_D  \vvvert _{\Omega}$.

The domain is meshed with triangles and divided in 8 identical squares. We use the FETI algorithm to solve the substructured problem with a Dirichlet preconditionner and the construction of statically admissible fields is performed using the Element Equilibration Technique (EET) \cite{Lad91, Rey2013}. For the record, the EET is two-step procedure: first, balanced traction fields are postprocessed on the edges of the elements, then, independent Neumann problems are solved on the elements with high precision (here each triangle is subdivided into $16$ elements).

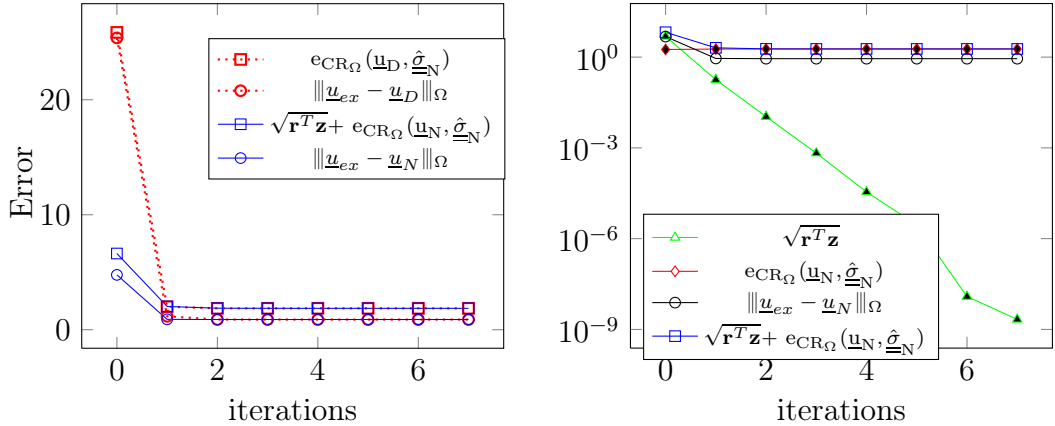
\begin{figure}[ht]
\begin{minipage}{.49\textwidth}
\begin{tikzpicture}
\begin{axis}[
width = 0.99\textwidth,
xlabel=iterations,
legend style={at={(0.3,0.9)}, anchor=north west}] 
\addplot [draw=red,thick,dotted, mark=square,mark options=solid] table[x=iterations,y=apf] {poutre_pol_bornes_EET.txt};
\addlegendentry{ {\scriptsize \ecr{{\un{u}}_D, \hat{\sig}_N}{\Omega}}}
\addplot [draw=red, thick,dotted,  mark=o,mark options=solid] table[x=iterations,y=evraie_uD] {poutre_pol_bornes_EET.txt};
\addlegendentry{{\scriptsize $\vvvert  \un{u}_{ex}-\un{u}_D  \vvvert _{\Omega}$ }}
\addplot [draw=blue, mark=square] table[x=iterations,y=somme] {poutre_pol_separation.txt};
\addlegendentry{{\scriptsize  $\sqrt{\res^T\bz}+$ \ecr{{\un{u}}_N, \hat{\sig}_N}{\Omega}}}
\addplot [draw=blue, mark=o] table[x=iterations,y=evraie_uN] {poutre_pol_separation.txt};
\addlegendentry{{\scriptsize $\vvvert  \un{u}_{ex}-\un{u}_N  \vvvert _{\Omega} $}}

\end{axis}
\draw (-0.8,2.2) node[scale=1.,rotate=90]{Error};
\end{tikzpicture}
\end{minipage}
\begin{minipage}{.49\textwidth}
 \begin{tikzpicture}
\begin{semilogyaxis}[
width = 0.99\textwidth,
xlabel=iterations,
legend style={at={(0.03,0.39)}, anchor=north west}] 
\addplot [draw=green, mark=triangle*] table[x=iterations,y=res] {poutre_pol_separation.txt};
\addlegendentry{ {\scriptsize $\sqrt{\res^T\bz}$ } }
\addplot [draw=red, mark=diamond*] table[x=iterations,y=EET] {poutre_pol_separation.txt};
\addlegendentry{ {\scriptsize \ecr{{\un{u}}_N, \hat{\sig}_N}{\Omega}}}
\addplot [draw=black, mark=o] table[x=iterations,y=evraie_uN] {poutre_pol_separation.txt};
\addlegendentry{{\scriptsize $\vvvert  \un{u}_{ex}-\un{u}_N  \vvvert _{\Omega} $}}
\addplot [draw=blue, mark=square] table[x=iterations,y=somme] {poutre_pol_separation.txt};
\addlegendentry{{\scriptsize  $\sqrt{\res^T\bz}+$ \ecr{{\un{u}}_N, \hat{\sig}_N}{\Omega}}}
\end{semilogyaxis}
\end{tikzpicture}
\end{minipage}
\caption{Error estimates}\label{fig:poutre_res}
\end{figure}

Figure~\ref{fig:poutre_res} illustrates the fast convergence of the estimators and the true errors against iterations.
 The quantities $\vvvert  \un{u}_{ex}-\un{u}_N  \vvvert _{\Omega}$ and $\vvvert  \un{u}_{ex}-\un{u}_D  \vvvert _{\Omega}$ are respectively bounded by $\sqrt{\res^T\bz}+ \ecr{{\un{u}}_N, \hat{\sig}_N}{\Omega}$ and $\ecr{{\un{u}}_D, \hat{\sig}_N}{\Omega}$. Obviously, when the solver has converged, the true errors $\vvvert  \un{u}_{ex}-\un{u}_N  \vvvert _{\Omega}$ and $\vvvert  \un{u}_{ex}-\un{u}_D  \vvvert _{\Omega}$ are equal and the estimates are the same.

 We observe that the algebraic error $\sqrt{\res^T\bz}$ regularly decreases to $10^{-9}$ along the iterations while the discretization error $\ecr{{\un{u}}_N, \hat{\sig}_N}{\Omega}$ is almost constant.

\subsection{Cracked structure}

Let us consider the structure of Figure~\ref{fig:Diez_pb} used in \cite{Par06}. We impose homogeneous Dirichlet boundary conditions in the bigger hole and on the base. The smaller hole is subjected to a constant unit pressure $p_0$. A unit traction force $\un{g}$ is applied normally to the surface on the left upper part. The other remaining boundaries are traction-free. A crack is also initiated from the smaller hole.
We mesh the structure with regular triangular linear elements and create 16 subdomains as shown in Figure~\ref{fig:Diez_decoupe}. We use the FETI algorithm to solve the substructured problem with a Dirichlet preconditioner. The construction of statically admissible fields is performed using the Element Equilibration Technique.

\begin{figure}[ht]
\begin{minipage}{.49\textwidth}
\centering
\includegraphics[width=.49\textwidth]{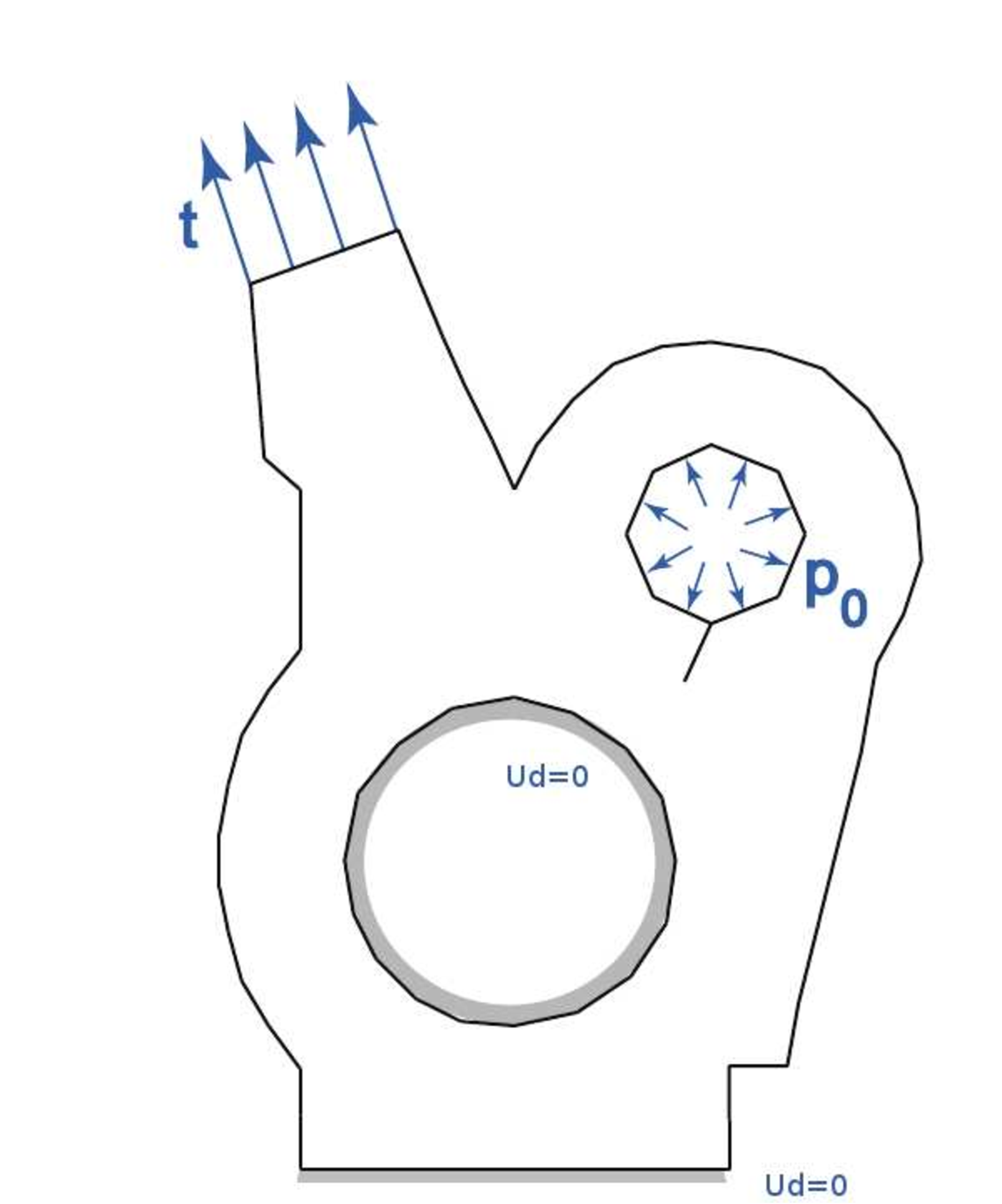}\caption{{\footnotesize Cracked structure model problem}}\label{fig:Diez_pb}
\end{minipage}
\begin{minipage}{.49\textwidth}
\centering
\includegraphics[width=.4\textwidth]{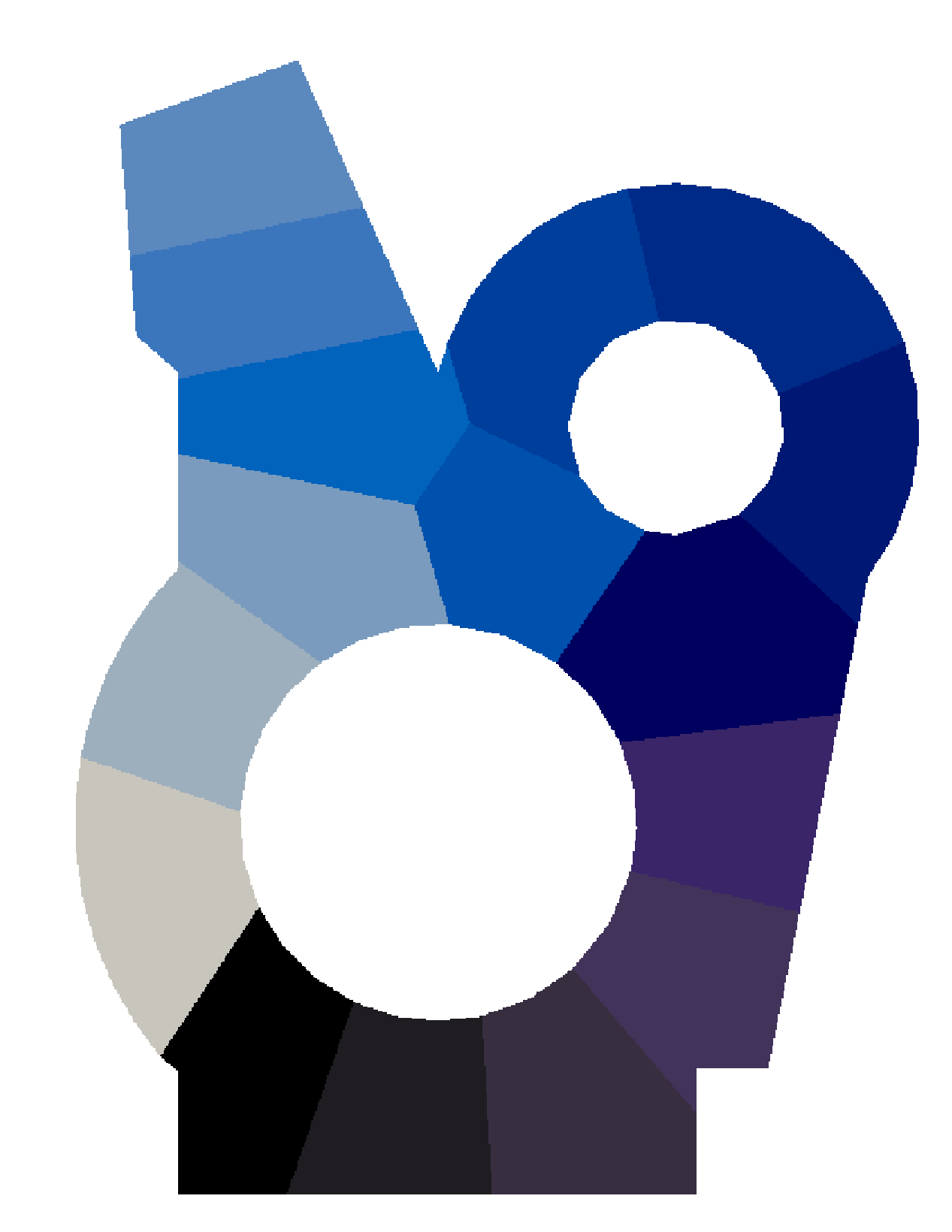}\caption{{\footnotesize Decomposition in 16 subdomains}}\label{fig:Diez_decoupe}
\end{minipage}
\end{figure}

\begin{figure}[ht]\label{fig:Diez_res}
\begin{minipage}{.49\textwidth}
\begin{tikzpicture}
\begin{semilogyaxis}[
width = 0.99\textwidth,
xlabel=iterations,
legend style={at={(0.3,.9)}, anchor=north west}] 
\addplot [draw=red, mark=diamond] table[x=iterations,y=apf] {diez_bornes_EET.txt};
\addlegendentry{ {\scriptsize \ecr{{\un{u}}_D, \hat{\sig}_N}{\Omega}}}
\addplot [draw=blue, mark=square] table[x=iterations,y=separe] {diez_bornes_EET.txt};
\addlegendentry{{\scriptsize  $\sqrt{\res^T\bz}+$ \ecr{{\un{u}}_N, \hat{\sig}_N}{\Omega}}}
\addplot [draw=black] table[x=iterations,y=apf_seq] {diez_bornes_EET.txt};
\addlegendentry{{\scriptsize  \ecr{{\un{u}}, \hat{\sig}}{\Omega} sequential}}
\end{semilogyaxis}
\end{tikzpicture}
\end{minipage}
\begin{minipage}{.49\textwidth}
\begin{tikzpicture}
\begin{semilogyaxis}[
width = 0.99\textwidth,
xlabel=iterations,
legend style={at={(0.06,0.35)}, anchor=north west}] 
\addplot [draw=green, mark=triangle*] table[x=iterations,y=res] {diez_separation.txt};
\addlegendentry{ {\scriptsize $\sqrt{\res^T\bz}$ } }
\addplot [draw=red, mark=diamond] table[x=iterations,y=EET] {diez_separation.txt};
\addlegendentry{{\scriptsize \ecr{{\un{u}}_N, \hat{\sig}_N}{\Omega}}}
\addplot [draw=blue, mark=square] table[x=iterations,y=somme] {diez_separation.txt};
\addlegendentry{{\scriptsize  $\sqrt{\res^T\bz}+$ \ecr{{\un{u}}_N, \hat{\sig}_N}{\Omega}}}
\end{semilogyaxis}
\end{tikzpicture}
\end{minipage}
\caption{Error estimates}\label{fig:diez_EET}
\end{figure}

As in \cite{Par10}, we observe on Figure~\ref{fig:diez_EET} L-shaped curves showing that the quality of the solution does not improve after the $8^{th}$ iteration where the estimators in the substructured case give results very comparable to a classical error estimation in sequential framework.

Again, the discretization error $\ecr{{\un{u}}_N, \hat{\sig}_N}{\Omega}$  does not change after the first iteration 
and the residual $\sqrt{\res^T\bz}$ becomes negligible with respect to the discretization error estimate much faster than the classical stopping criterion tells.

Figure~\ref{fig:map_error} enables us to verify that the map of element contribution to the error estimator is converged when the algebraic error is negligible  with respect to the discretization error (10 times smaller): the maps represent the difference between the final element contributions (iteration 20) and the current element contributions at iterations 1 (initialization) and 7 (when algebraic error is negligible).

\begin{figure}[ht]
\centering
\subfigure[ {\scriptsize $\frac{\left| \ecr{{\un{u}}_N, \hat{\sig}_N}{E}_{it=1}- \ecr{{\un{u}}_N, \hat{\sig}_N}{E}_{it=20} \right|}{\ecr{{\un{u}}_N, \hat{\sig}_N}{\Omega}_{it=20}}$}]{ \centering
   \includegraphics[width=0.3\textwidth] {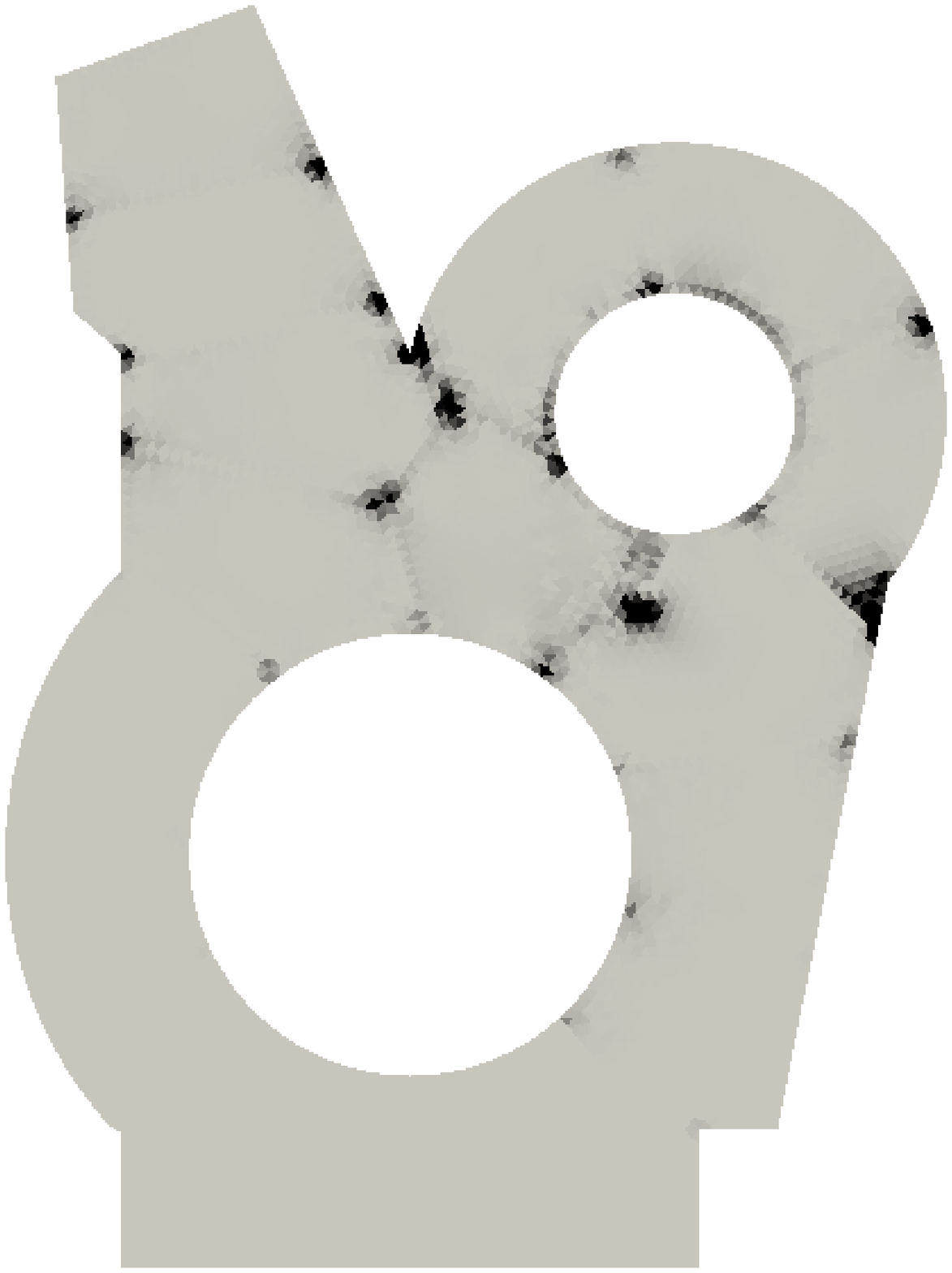}
   \label{fig:subfig1} }
 \subfigure[ {\scriptsize $\frac{\left| \ecr{{\un{u}}_N, \hat{\sig}_N}{E}_{it=7}- \ecr{{\un{u}}_N, \hat{\sig}_N}{E}_{it=20} \right|}{\ecr{{\un{u}}_N, \hat{\sig}_N}{\Omega}_{it=20}}$}]{\centering
   \includegraphics[width=0.3\textwidth] {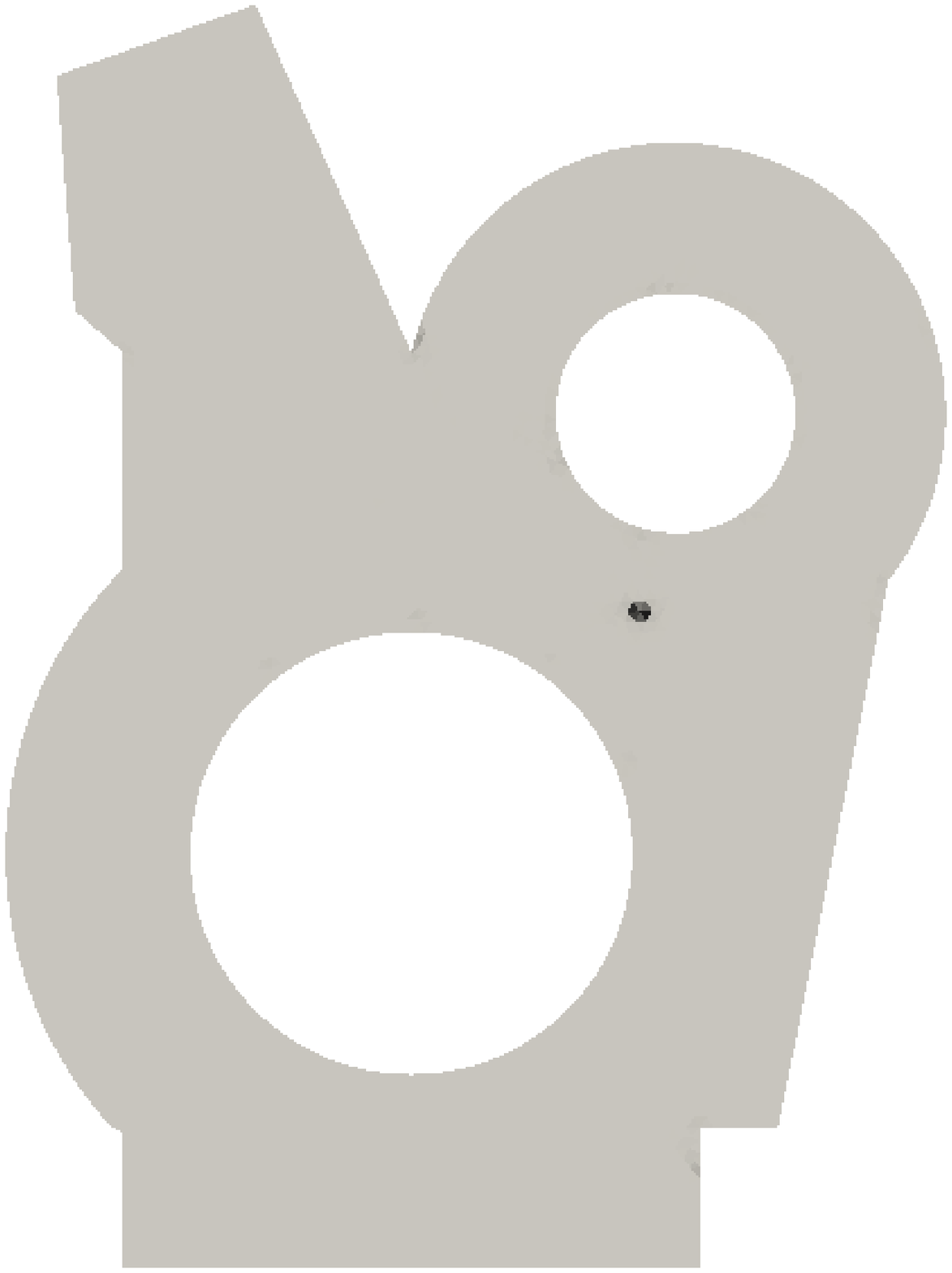}
   \label{fig:subfig2}}
 \subfigure[ ]{\centering
   \includegraphics[width=0.1\textwidth] {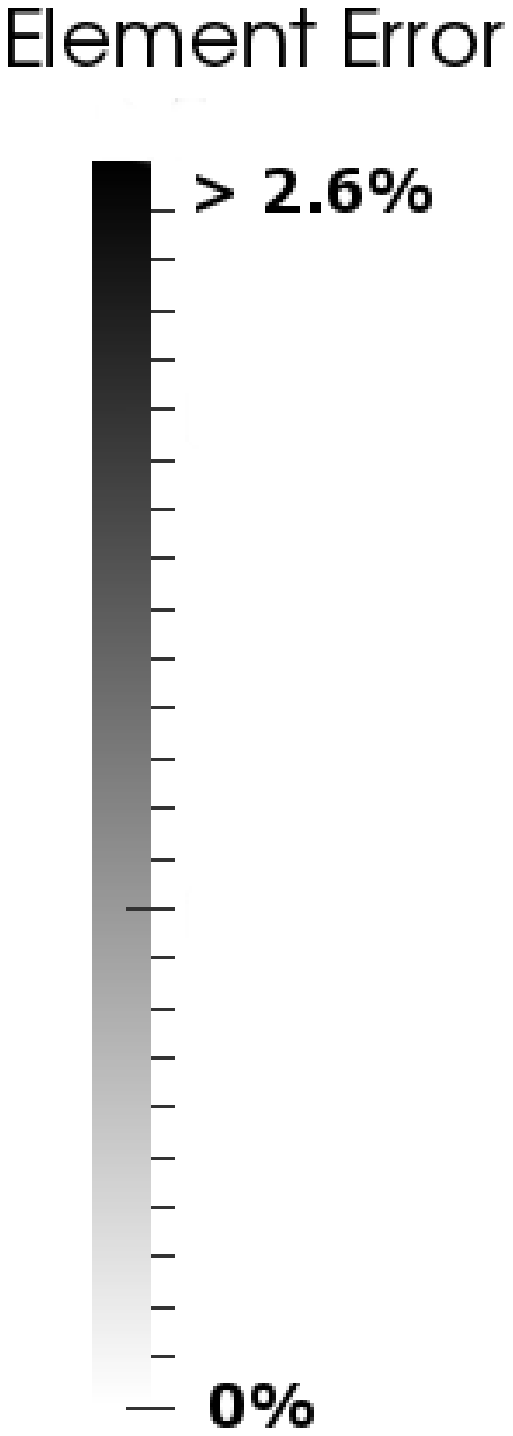}
   \label{fig:subfig3}}
\caption{Convergence of the element contribution to the estimator}\label{fig:map_error}
\end{figure}

\section{Conclusion}
\label{sec:conclusion}
This paper introduces a new guaranteed upper bound of the error in the framework of non overlapping decomposition method. This upper bound is the sum of two terms: one term is exactly the algebraic error (error due to the use of an iterative solver) and the second term is mostly due to the discretization error.

From a practical point of view, the evaluation of the algebraic error is trivial. We just compute the norm associated to the preconditioner of FETI or BDD. The evaluation of the second term of the inequality relies on the capacity to build a statically admissible stress field from the displacement fields resulting from Neumann problems per subdomains. The examples show that this quantity does not evolve much after the first iteration and then represents the discretization error. 

This separation offers the possibility to define a new stopping criterion for the iterative solver based on the non-improvement of the global quality of the approximation.
If a better quality of the solution is required, the error estimator provides an error card that can guide the remeshing operation. 

\bibliography{Biblio}

\end{document}